\documentclass{article}

\usepackage[utf8]{inputenc}
\usepackage{a4wide}
\usepackage{amsmath}
\usepackage{amsfonts}
\usepackage{amssymb}
\usepackage{mathtools}
\usepackage{xcolor}
\usepackage{enumerate}
\usepackage{natbib}
\usepackage{xspace}
\usepackage{subcaption}
\usepackage{xparse}
\usepackage[thmmarks, amsmath, thref, hyperref]{ntheorem}

\usepackage[pdfauthor={Sergio Waitman},%
            pdftitle={Incremental stability of Lur'e systems through piecewise-affine approximations},%
            pdftex,bookmarks]{hyperref}

\definecolor{dred}{rgb}{0.5,0,0}
\definecolor{red}{rgb}{0.8,0,0}
\definecolor{blue}{rgb}{0,0,0.5}
\definecolor{green}{rgb}{0,0.3,0}
\definecolor{grey}{rgb}{0.5,0.5,0.5}
		            
\usepackage{redef}
\newcommand{\paper}{paper\xspace}

\newcommand{\vpwa}{\varphi_\textnormal{PWA}}
\newcommand{\etaref}{\eta_\textnormal{ref}}

\def\smath#1{\text{\scalebox{1.1}{$#1$}}}

\newcommand{\qedal}{\tag*{$\Box$}}

\newtheorem{theorem}{Theorem}[section]
\newtheorem{definition}[theorem]{Definition}

\newtheorem{proposition}[theorem]{Proposition}
\newtheorem{assumption}[theorem]{Assumption}
\newtheorem{algorithm}[theorem]{Algorithm}
\theorembodyfont{\upshape}
\newtheorem{example}{Example}

\theoremstyle{nonumberplain}
\theoremsymbol{\ensuremath{\blacksquare}}
\theorembodyfont{\upshape}
\newtheorem{proof}{Proof}

\begin{document}

\title{Incremental stability of Lur'e systems through piecewise-affine approximations}

\author{S.~Waitman,
L.~Bako,
P.~Massioni,
G.~Scorletti
and V.~Fromion}

\maketitle

\begin{abstract}
	Lur'e-type nonlinear systems are virtually ubiquitous in applied control theory, which explains the great interest they have attracted throughout the years. The purpose of this \paper is to propose conditions to assess incremental asymptotic stability of Lur'e systems that are less conservative than those obtained with the incremental circle criterion. The method is based on the approximation of the nonlinearity by a piecewise-affine function. The Lur'e system can then be rewritten as a so-called piecewise-affine Lur'e system, for which sufficient conditions for asymptotic incremental stability are provided. These conditions are expressed as linear matrix inequalities (LMIs) allowing the construction of a continuous piecewise-quadratic incremental Lyapunov function, which can be efficiently solved numerically. The results are illustrated with numerical examples.
\end{abstract}



\section{Introduction}

The so-called Lur'e-type nonlinear systems, given by the feedback interconnection of a linear time-invariant (LTI) system and a memoryless nonlinearity $\varphi$, represent an important class of systems with practical application in virtually any domain of system theory. The study of these systems is closely connected with the development of the absolute stability problem (see e.g.~\cite{Liberzon2006}), which consists in establishing conditions to ensure asymptotic stability of the origin for a set of nonlinear functions in a sector.

In this \paper, we are interested in assessing incremental stability of Lur'e systems, i.e. the stability of every system trajectory with respect to each other. Several different notions of incremental stability coexist~\citep[see e.g.][]{Fromion1997,Lohmiller1998,Angeli2002,Pavlov2004}, but all have in common the fact that they ensure strong qualitative properties on the system behavior, such as asymptotic independence of initial conditions and the unicity of the steady state. For this reason, incremental stability is often used to cope with problems involving tracking/synchronization and anti-windup control~\citep[see e.g.][]{Rantzer2000a,Kim2015}.

In the framework of input-output stability, \cite{Zames1966} proposed graphical conditions to ensure (incremental) stability of Lur'e systems, known as the (incremental) circle criterion. These conditions are established for nonlinearities belonging to a sector and, in this sense, the nonlinearity can be seen as a bounded perturbation on the linear dynamics of the system. The description via sector bounds yields stability results that tend to be quite conservative, as the sector bound gives a very crude representation of the nonlinear operator. For stability analysis, an attempt to reduce the conservatism was made by transforming the feedback loop via the addition of so-called \emph{Popov-Zames-Falb frequency-dependent multipliers}~\citep{Zames1966,Zames1968a}. However, it turns out that this approach is not applicable when incremental stability is considered~\citep{Kulkarni2002}. \cite{Fromion2004a} showed that there exist Lur'e nonlinear systems for which multiplier-based analysis ensures finite gain stability, but which are not incrementally stable. On the other hand, necessary and sufficient conditions for incremental stability of Lur'e systems were proposed by~\cite{Fromion2003a}, but with the drawback of being NP-hard. There is then a need for an alternative approach to the assessment of incremental stability of Lur'e systems, which is less conservative than the celebrated incremental circle criterion while being efficiently solvable. For this reason, we consider the analysis via piecewise-affine approximations.

Piecewise-affine (PWA) systems are nonlinear systems described by piecewise-affine differential equations. They can be used to naturally describe systems containing piecewise-affine nonlinearities (such as saturations, relays and dead zones), or as an approximation of more general nonlinear systems. The interest in this class of systems lies in the fact that their description is quite close to that of LTI systems, allowing transposition of classic results on stability and performance analysis while being able to present quite complex nonlinear dynamics. \cite{Johansson1998} introduced piecewise-quadratic Lyapunov functions to the analysis of PWA systems through the use of the \Sproc. The approach was extended to consider the analysis of incremental properties of PWA systems by~\cite{Waitman2016}.

In this \paper, we propose a method to assess incremental stability of Lur'e systems through piecewise-affine approximations. The nonlinearity $\varphi$ is replaced by a piecewise-affine function $\vpwa$ plus an approximation error $\epsilon$, which is characterized by its Lipschitz constant. This allows us to rewrite the Lur'e system as the interconnection of a PWA system with the approximation error, in what we may call a PWA Lur'e system (see Fig.~\ref{fig:AppBlckDgm}). Through the refinement of $\vpwa$, we are able to control the approximation error, and hence expect to obtain less conservative results. In this sense, we address two technical questions: obtain sufficient conditions to assess incremental stability of PWA Lur'e systems; and propose a method allowing the construction of $\vpwa$. Although techniques to construct piecewise-affine approximations exist in the literature~\citep[see e.g.][]{Zavieh2013,Azuma2010}, we introduce an approximation method ensuring a given upper bound on the Lipschitz constant of the approximation error.
\begin{figure}
	\centering
		\includegraphics[width=0.5\linewidth]{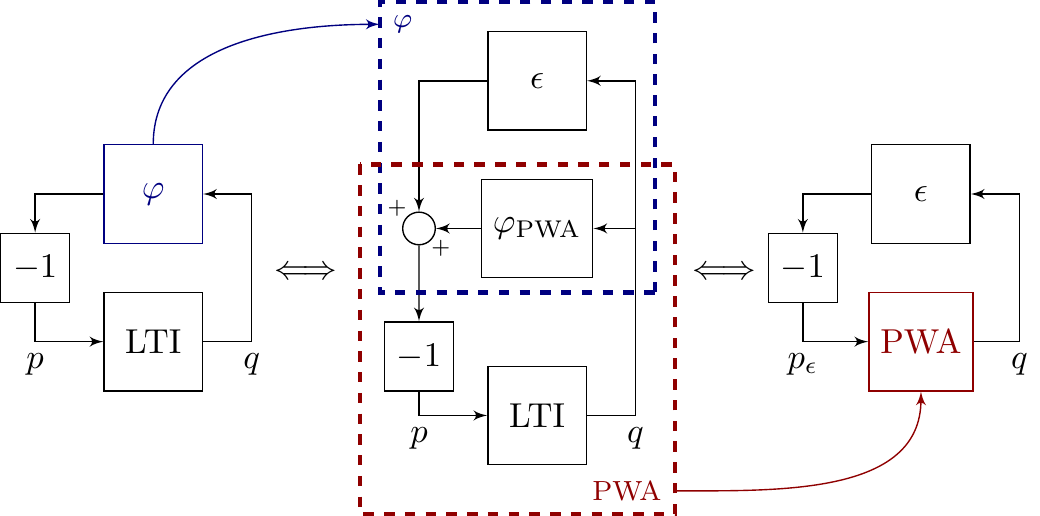}
	\caption{Block diagram illustrating the approach taken in this \paper.}
	\label{fig:AppBlckDgm}
\end{figure}

The \paper is organized as follows. Section~\ref{se:ProbForm} states the problem of ensuring incremental asymptotic stability of Lur'e systems. The proposed approach is presented in Section~\ref{se:PropApp}. In Section~\ref{se:IncStabPWALure}, sufficient conditions for incremental asymptotic stability of PWA Lur'e systems are presented. Section~\ref{se:AppNL} proposes a method to construct $\vpwa$ that ensures an upper bound on the Lipschitz constant of the approximation. Finally, Section~\ref{se:NumEx} contains numerical examples illustrating the results obtained with the proposed approach.

\subsection*{Notation} We denote by $\norm{\cdot}$ the Euclidean norm. The real half line $[0,+\infty)$ is denoted by $\R_+$. The interior of a set $\mathcal{A}$ is denoted $\inte{\mathcal{A}}$. For a vector $v = [v_1,\ldots,v_n] \in \R^n$, $v \succ 0$ (resp. $v \succeq 0$) is equivalent to the componentwise inequality $v_i > 0$ (resp. $v_i \geq 0$), $\forall	i \in \{1,\ldots,n\}$. For a matrix $A \in \R^{n\times n}$, $A \succ 0$ (resp. $A \succeq 0$) denotes that $A$ is positive definite (resp. semi-definite). The symbol $\bullet$ replaces the corresponding symmetric block in a symmetric matrix. The column concatenation of two matrices $A$ and $B$ of compatible dimensions, denoted by $\col$, is such that $\col(A,B) = {\text{\footnotesize$\nmatrix{c}{A\\B}$}}$.

The function $\phi: \R_+ \times \R_+ \times X \rightarrow X$ is called the \emph{state transition map} and is such that $x = \phi(t,t_0,x_0)$ is the state $x \in X$ attained at instant $t$ when the system evolves from $\xo \in X$ at the instant $t_0$.

A function $\rho: \R_+ \rightarrow \R_+$ is said to be positive definite if it is such that $\rho(0) = 0$ and $\rho(r) >0$, $\forall r \neq 0$. We denote by $\K$ the class of continuous and strictly increasing functions $\alpha: \R_+ \rightarrow \R_+$ for which $\alpha(0) = 0$. A function $\alpha$ is of class $\Koo$ if it is of class $\K$ and unbounded. A continuous function $\beta: \R_+\times\R_+ \rightarrow \R_+$ is of class $\KL$ if for any fixed $t \geq 0$, $\beta(\cdot,t) \in \K$ and, for fixed s, $\beta(s,\cdot)$ is decreasing with $\lim_{t\rightarrow\infty} \beta(s,t) = 0$.

\section{Problem formulation}
\label{se:ProbForm}

In this \paper, we are interested in establishing conditions to assess the incremental asymptotic stability of nonlinear Lur'e systems given by
\begin{equation}
\label{eq:Lure}
	\left\{ 	
 	\begin{aligned}
 		\dxt &= A\xt + B\pt \\
 		\qt &= C\xt \\
 		\pt &= -\varphi(\qt) \\
 		x(0) &= x_0 		
 	\end{aligned}
 	\right.
\end{equation}
where $x(t)\in X \subseteq \R^n$ is the state, $\pt, \qt \in \R$ are internal signals and $\varphi$ is a given memoryless Lipschitz nonlinearity with $\varphi(0) = 0$. Let us recall the following definition, adapted from~\cite{Angeli2002}.

\begin{definition}
\label{def:dAS}
	We say that system~\eqref{eq:Lure} is incrementally asymptotically stable if there exists a function $\beta$ of class $\KL$ so that for all $x_0, \tx_0 \in X$ and all $t \geq 0$ the following holds
	\begin{equation}
	\label{eq:dGAS}
		\norm{x(t) - \tx(t)} \leq \beta(\norm{x_0 - \tx_0}, t)
	\end{equation}
	with $\xt = \phi(t,0,x_0)$ and $\txt = \phi(t,0,\tx_0)$. If $X = \R^n$, the system is said to be incrementally globally asymptotically stable.
\end{definition}

Parallel to standard stability conditions, incremental asymptotic stability may be shown to be equivalent to a Lyapunov-like condition. In view of the adapted definition adopted in this \paper, let us recall the following theorem, adapted from~\cite{Angeli2002}.

\begin{theorem}
\label{th:dAS}
	System~\eqref{eq:Lure} is incrementally asymptotically stable as in Definition~\ref{def:dAS} if there exist a continuous function $V: X \times X \rightarrow \R_+$, called an incremental Lyapunov function, and $\Koo$ functions $\alpha_1$ and $\alpha_2$ such that
	\begin{equation}
	\label{eq:dASnorm}
		\alpha_1\big(\!\norm{x - \tilde{x}}\!\big) \leq V(x,\tilde{x}) \leq \alpha_2\big(\!\norm{x - \tilde{x}}\!\big)
	\end{equation}
	for every $x,\tx \in X$, and along any two trajectories $x,\tx$, starting respectively from $x_0,\tx_0 \in X$, $V$ satisfies for any $t \geq 0$
	\begin{equation}
	\label{eq:dASnegdef}
	 	V(\xt,\txt) - V(x_0,\tx_0) \leq -\int_0^t\!\rho\big(\!\norm{x(\tau) - \tx(\tau)}\!\big) \, d\tau
	\end{equation}
	with $\xt = \phi(t,0,\xo)$, $\txt = \phi(t,0,\txo)$ and $\rho$ a positive definite function.
\end{theorem}

\section{Proposed approach}
\label{se:PropApp}

The traditional approach to assess incremental stability of Lur'e systems~\eqref{eq:Lure} is to use the incremental circle criterion (see e.g.~\cite{Zames1966,Fromion1999}). This involves embedding $\varphi$ in a so-called incremental sector.

\begin{definition}
\label{def:IncSector}
	The nonlinearity $\varphi$ is said to belong to the incremental sector $[\kappa_1,\kappa_2]$ if $\kappa_1 \leq (\varphi(q) - \varphi(\tq))/(q-\tq)) \leq \kappa_2$, for all $q,\tq \in \R$, with $q \neq \tq$.
\end{definition}

From Definition~\ref{def:IncSector}, it is clear that a Lipschitz nonlinearity $\varphi$, with Lipschitz constant $L$, belongs to the sector $[-L,L]$. The incremental circle criterion gives conditions to assess incremental stability of \emph{every} nonlinearity inside an incremental sector. By doing so, we obtain tractable conditions to perform the analysis, but at the price of some conservatism. This is due to the fact that, in general, incremental sector conditions provide a very crude description of $\varphi$. To cope with this problem, we propose computing a piecewise-affine approximation $\vpwa$ of the nonlinearity $\varphi$, so that~\eqref{eq:Lure} is transformed into the interconnection of a PWA system with the approximation error:
	\begin{equation}
	\label{eq:LurePWAApp}
		\left\{
	 		\begin{aligned} &
		 		\begin{aligned}
		 			\dxt &= A_i\xt + a_i + Bp_\epsilon(t) \\
	 				\qt &= C_i\xt + c_i + Dp_\epsilon(t)
	 			\end{aligned}   \quad \text{for } \xt \in X_i \\
 				& p_\epsilon(t) = -\epsilon(\qt) \\
 				& x(0) = x_0
	 		\end{aligned}
	 	\right.
	\end{equation}
We shall refer to~\eqref{eq:LurePWAApp} as a \emph{PWA Lur'e system}. We make the assumption that the approximation error $\epsilon$ is Lipschitz with Lipschitz constant $\eta$. The regions $X_i$, for $i \in \I := \{1, \ldots, N\}$, are closed convex polyhedral sets $X_i = \{ x \in X \mid G_ix + g_i \succeq 0 \}$ with non-empty and pairwise disjoint interiors such that $\bigcup_{i\in \I} X_i= X$. Then, $\{X_i\}_{i \in \I}$ constitutes a finite partition of $X$. From the geometry of $X_i$, the intersection $X_i \cap X_j$ between two different regions is always contained in a hyperplane, i.e. $X_i \cap X_j \subseteq \left\{ x \in X \mid E_\ij x + e_\ij = 0\right\}$. The approach is illustrated in Fig.~\ref{fig:AppBlckDgm}, and formalized in the next proposition.

\begin{proposition}
\label{prop:LurePWA}
	Let $\mRi \subset \R$, $i \in \I = \lbrace 1, \ldots, N\rbrace$, be non-empty intervals with pairwise disjoint interiors, such that $\{\mRi\}_{i \in \I}$ forms a partition of $\R$. Let the scalar nonlinearity $\varphi$ in~\eqref{eq:Lure} be decomposed as $\varphi(q) = \vpwa(q) + \epsilon(q)$, with $\vpwa$ a piecewise-affine function given by $\vpwa(q) = r_iq + s_i$, for $q \in \mRi$. Then, the Lur'e system~\eqref{eq:Lure} is equivalent to the PWA Lur'e system~\eqref{eq:LurePWAApp}, with $\epsilon(q) := \varphi(q) - \vpwa(q)$, $A_i := A - r_iBC$, $a_i := -s_iB$, $C_i = C$, $c_i = 0$, $D = 0$ and $X_i = \{x \in X \mid Cx \in \mRi\}$.
\end{proposition}
\begin{proof}
	The proof follows after straightforward manipulations. Indeed, it suffices to replace $\varphi(q)$ by the sum $\vpwa(q) + \epsilon(q)$. Then, using the fact that $\vpwa(q) = r_iq + s_i = r_iCx + s_i$, the nonlinear system~\eqref{eq:Lure} may be rewritten as
\begin{align}
	\dx &= Ax - B(\vpwa(q) + \epsilon(q)) \nonumber \\
	&= Ax - B(r_iCx + s_i + \epsilon(q)) \nonumber \\
	&= (A - r_iBC)x - s_iB - B\epsilon(q)  \\
	&=: A_ix + a_i + Bp_\epsilon \qedal 
\end{align}
\end{proof} 

By performing analysis on~\eqref{eq:LurePWAApp}, we replace the test for every $\varphi \in \lbrace \varphi \mid \varphi \in [-L,L] \rbrace$ by the test for every $\varphi \in \lbrace \varphi \mid \varphi = \vpwa + \epsilon, \text{ with } \epsilon \in [-\eta,\eta]\rbrace$. As we are able to control the approximation error through the refinement of $\vpwa$ (and thus to control $\eta$), this allows us to obtain a PWA Lur'e system whose nonlinearity is described by much tighter sector bounds (see Fig.~\ref{fig:SectorBounds}). Hence, the analysis provides potentially less conservative results for the incremental analysis of Lur'e systems. The approach is presented in the next algorithm.
\begin{figure}
	\centering
		\includegraphics[width=0.7\linewidth]{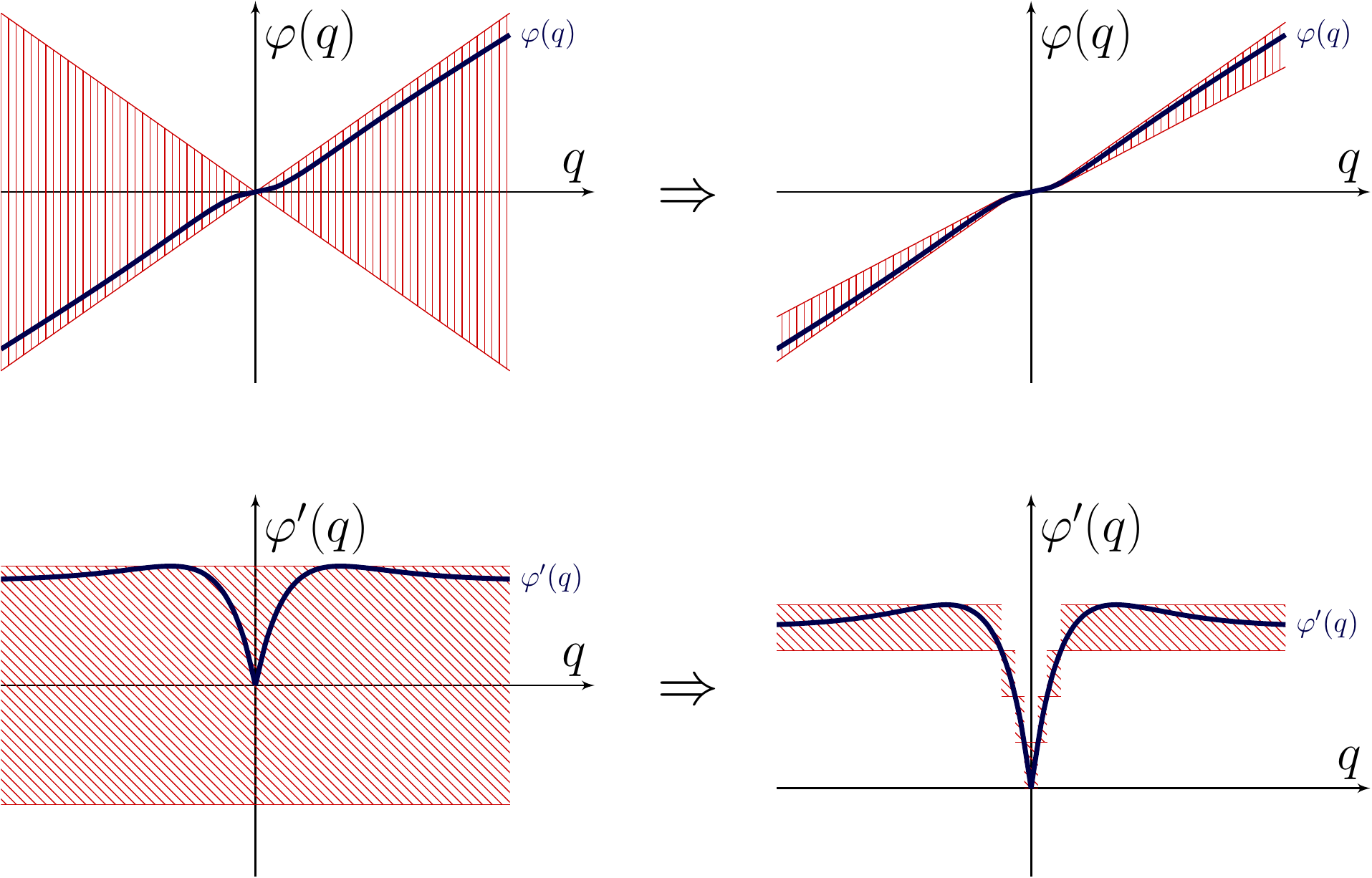}
	\caption{Comparison between the sectors describing the nonlinearity $\varphi$ for the incremental circle criterion (left) and the piecewise-affine approach (right).}
	\label{fig:SectorBounds}
\end{figure}

\begin{algorithm}
\label{alg:Method}
	Given a Lur'e system~\eqref{eq:Lure} with a memoryless Lipschitz nonlinearity $\varphi$:
	\begin{enumerate}
		\item Compute a piecewise-affine approximation $\vpwa$ so that $\epsilon = \varphi - \vpwa$ is Lipschitz, with a Lipschitz constant $\eta$ smaller than a given upper bound $\etaref$.
		\item Use Proposition~\ref{prop:LurePWA} to construct an equivalent PWA Lur'e system~\eqref{eq:LurePWAApp} from~\eqref{eq:Lure}.
		\item Assess incremental asymptotic stability of~\eqref{eq:LurePWAApp}, and, if positive, conclude on the incremental asymptotic stability of~\eqref{eq:Lure}.
	\end{enumerate}
\end{algorithm}

To apply Algorithm~\ref{alg:Method}, we need to consider two questions: how to assess incremental asymptotic stability of PWA Lur'e systems, and how to construct piecewise-affine approximations ensuring an upper bound on the Lipschitz constant of the approximation error (and thus on its incremental sector bounds). These problems shall be addressed in the next sections.

\section{Incremental stability of PWA Lur'e systems}
\label{se:IncStabPWALure}

In this section we propose conditions to assess incremental asymptotic stability of PWA Lur'e systems given by~\eqref{eq:LurePWAApp}. The results are based on the construction of a piecewise-quadratic incremental Lyapunov function and application of Theorem~\ref{th:LurePWA_dAS}. 

When studying incremental properties, it is standard to consider a fictitious augmented system (see e.g.~\cite{Angeli2002,Fromion1997}). Considering the PWA structure of~\eqref{eq:LurePWAApp}, we can define an augmented system given by
\begin{equation}
\label{eq:PWAsysUYaug}
	\left\lbrace 
	\begin{aligned} 
		&\begin{aligned}
			\dbxt &= \bA_\ij\bxt + \bB\bpt \\
			\bqt &= \bC_\ij\bxt + \bD\bpt
		\end{aligned} & \text{for } \bxt\in X_\ij\\ 
		&\bx(0) = \bx_0
	\end{aligned}\right.
\end{equation}
where $\bx = \col(x,\tx,1)$, $\bp = \col(p_\epsilon,\tp_\epsilon)$, $\bq = q - \tq$ and
\begin{equation}
\label{eq:AugMat}
\begin{aligned}
	\bA_\ij &= \nmatrix{ccc}{A_{i} & 0 & a_{i} \\ 0 & A_{j} & a_{j} \\ 0 & 0 & 0} &
	\bB &= \nmatrix{cc}{B & 0 \\ 0 & B \\ 0 & 0}
	\\
	\bC &= \nmatrix{ccc}{C & -C & \phantom{-}0}
	&
	\bD &= \nmatrix{cc}{D & -D}.
\end{aligned}
\end{equation}

The space $\bX$ is defined as $\bX = X \times X \times \{1\}$, and regions $X_\ij$ are defined as $X_\ij = \{\bx \in \bX \mid x \in X_i \text{ and } \tx \in X_j\}$. Each region $X_\ij$ is described by $X_\ij = \{ \bx \in \bX \mid \bG_\ij \bx \succeq 0 \}$ where
\begin{equation}
 	\bG_\ij =  \nmatrix{ccc}{G_{i} & 0 & g_{i} \\ 0 & G_{j} & g_{j}}.
\end{equation}

Analogously to the state partition $\{X_i\}_{i \in \I}$ of system~\eqref{eq:LurePWAApp}, the intersection between any two regions $X_\ij$ and $X_{kl}$ of~\eqref{eq:PWAsysUYaug} is either empty or contained in the hyperplane given by
\begin{equation}
\label{eq:AugSysHypP}
	X_\ij \cap X_{kl} \subseteq \left\{ \bx \in \bX \mid \bE_{ijkl} \bx = 0\right\}.
\end{equation}

We shall propose conditions to compute an incremental Lyapunov function possessing the following piecewise-quadratic structure:
\begin{equation}
\label{eq:PWQStorFunc}
	V(x,\tx) = \begin{cases}
 			(x - \tx)^TP_i(x - \tx) & \!\text{for } \bx \in X_{ii} \\
 			\bx^T\bP_\ij\bx & \!\text{for } \bx \in X_\ij, \, i \neq j
 		\end{cases}
\end{equation}
As presented in~\cite{Waitman2016}, the choice of a quadratic function on $(x - \tx)$ on regions $X_{ii}$ does not lead to any loss of generality. Indeed, it is a consequence of the fact that $V(x,x) = 0$, for every $x \in X$, due to~\eqref{eq:dASnorm}.

Let us denote by $I_n$ the $n \times n$ identity matrix, and let $\bI_n \in \R^{2n\times 2n}$ and $\bJ_n \in \R^{(2n+1)\times (2n+1)}$ denote the following matrices
\begin{align}
 	\bI_n &= \nmatrix{cc}{I_n & -I_n \\ -I_n & I_n} & 
 	\bJ_n &= \nmatrix{ccc}{I_n & -I_n & 0 \\ -I_n & I_n & 0 \\ 0\TBS & 0 & 0}.
\end{align}
We are then able to state the following theorem.

\begin{theorem}
\label{th:LurePWA_dAS}
	Let~\eqref{eq:LurePWAApp} be a PWA Lur'e system, and let $\epsilon$ be Lipschitz continuous with Lipschitz constant $\eta > 0$. If there exist symmetric matrices $P_{i} \in \R^{n \times n}$ and $\bP_{ij} \in \R^{(2n+1) \times (2n+1)}$; $U_{ij}$, $R_{ij}$, $W_{ij} \in \R^{p_{ij} \times p_{ij}}$ with nonnegative coefficients and zero diagonal; $L_{ijkl} \in \R^{(2n+1)\times 1}$ and positive scalars $\sigma_1, \sigma_2, \sigma_3$ such that
	\begin{equation}
	\label{eq:ContThm1}
		\begin{cases}
			P_{i} - \sigma_1I_n \succeq 0 \\
			P_{i} - \sigma_2I_n \preceq 0 \\
			\nmatrix{cc}{A_i^TP_{i} + P_{i}A_i + C^TC + \sigma_3I_n & P_iB + C^TD \\ \bullet & D^TD- \eta^{-2}I_p} \preceq 0
		\end{cases}
	\end{equation}
	for $i \in \I$,
	\begin{equation}
	\label{eq:ContThm2} 
		\begin{cases}
			\bP_{ij} - \sigma_1\bJ_n - \bE_{ij}^TU_{ij}\bE_{ij} \succeq 0 \\
			\bP_{ij} - \sigma_2\bJ_n + \bE_{ij}^TR_{ij}\bE_{ij} \preceq 0 \\
			\nmatrix{cc}{
				\left({ \begin{smallmatrix} 
      				\smath{\bA_\ij^T\overline{P}_\ij + \overline{P}_\ij \bA_\ij + \bC_\ij^T\bC_\ij + {}} \\ \smath{\sigma_3\bJ_n + \bE_{ij}^TW_{ij}\bE_{ij}}
    			\end{smallmatrix}}\right) & 
				\bP_\ij\bB + \bC_\ij^T\bD \\
				\bullet &
				\bD^T\bD- \eta^{-2}\overline{I}_p} \preceq 0
		\end{cases}
	\end{equation}
	for $(i,j) \in \I \times \I$, $i \neq j$, and
	\begin{equation}
	\label{eq:ContThm3}
		\bP_{ij} = \bP_{kl} + L_{ijkl}\bE_{ijkl} + \bE_{ijkl}^TL_{ijkl}^T
	\end{equation}
	for $(i,j),(k,l)$ such that $X_\ij \cap X_{kl} \neq \varnothing$ are satisfied, then the Lur'e system~\eqref{eq:LurePWAApp} is incrementally asymptotically stable.
\end{theorem}

\begin{proof}
	According to Theorem~\ref{th:dAS}, \eqref{eq:LurePWAApp} is incrementally asymptotically stable if there exists a continuous incremental Lyapunov function $V$, which is lower and upper bounded by class $\Koo$ functions, and respects the integral constraint~\eqref{eq:dASnegdef}. We shall prove the theorem by showing that feasibility of~\eqref{eq:ContThm1}--\eqref{eq:ContThm3} implies the existence of such a function possessing the structure~\eqref{eq:PWQStorFunc}.
	
	\noindent\emph{Continuity} -
	We first show that $V$ is a continuous function of $\bx$. This is clearly the case inside every cell, so we just need to show continuity on the boundaries. From~\eqref{eq:AugSysHypP}, $\bE_{ijkl}\bx = 0$ for all $\bx \in X_\ij \cap X_{kl}$, then~\eqref{eq:ContThm3} implies that $ \bx^T\bP_\ij\bx = \bx^T\bP_{kl}\bx$ for $\bx \in X_\ij \cap X_{kl}$ and hence that $V$ is continuous.
	
	\noindent\emph{Norm bounds} -
	The first inequality in~\eqref{eq:ContThm2}, post and pre multiplied respectively by $\bx$ and $\bx^T$, implies that $\bx^T\bP_\ij\bx - \sigma_1\norm{x - \tx}^2 \geq \bx^T\bG_{ij}^TU_{ij}\bG_{ij}\bx$. Since $U_{ij}$ is composed of nonnegative coefficients, the right-hand side of the previous inequality is nonnegative whenever $\bx \in X_\ij$. This implies that
	\begin{equation}
	\label{eq:PWQPosDefXij_dAS}
	 	\bx^T\bP_\ij\bx  \geq \sigma_1\norm{x - \tx}^2 \qquad \text{for } \bx \in X_{ij}.
	\end{equation}
	
	The first inequality in~\eqref{eq:ContThm1} implies that $V(x,\tx) \geq \sigma_1\norm{x - \tx}^2$ for all $\bx \in X_{ii}$. With~\eqref{eq:PWQPosDefXij_dAS}, this guarantees that	
	\begin{equation}
	\label{eq:PWQPosDef_dAS}
		V(x,\tx) \geq \sigma_1\norm{x - \tx}^2, \quad\forall x,\tx \in X.
	\end{equation}
	
	Proceeding exactly as before, the second inequalities in~\eqref{eq:ContThm1} and~\eqref{eq:ContThm2} imply that
	\begin{equation}
	\label{eq:PWQUpBound_dAS}
		V(x,\tx) \leq \sigma_2\norm{x - \tx}^2, \quad\forall x,\tx \in X.
	\end{equation}
	Inequalities~\eqref{eq:PWQPosDef_dAS} and~\eqref{eq:PWQUpBound_dAS} imply that the continuous piecewise quadratic function $V$ given by~\eqref{eq:PWQStorFunc} is such that
	\begin{equation}
	 	\sigma_1\norm{x - \tx}^2 \leq V(x,\tx) \leq \sigma_2\norm{x - \tx}^2.
	\end{equation}
		
	\noindent\emph{Integral constraint} -
	We now show that the incremental Lyapunov function respects the integral constraint~\eqref{eq:dASnegdef}. Using the same arguments as before, the last inequality in~\eqref{eq:ContThm2}, post and pre multiplied by $\col(\bx, \bp)^T$ and $\col(\bx, \bp)$, implies that
	\begin{multline}
		\bx^T\bP_\ij(\bA_\ij\bx + \bB\bp)  + (\bA_\ij\bx + \bB\bp)^T\bP_\ij\bx \;+ \\ 
		(\bC_\ij\bx + \bD\bp)^T(\bC_\ij\bx + \bD\bp) - \eta^{-2}\bp^T\overline{I}_p\bp  \leq -\sigma_3\norm{x - \tx}^2
	\end{multline}
	for all $\bp \in \R^2$ and all $\bx \in X_\ij$. Let $t_a$ and $t_b$ be two time instants such that the state trajectory of system~\eqref{eq:PWAsysUYaug} remains in $\Xij$ on the interval $[t_a,t_b]$. By noticing that $\dbx = \bA_\ij\bx + \bB\bp$, and integrating from $t_a$ to $t_b$ along trajectories of~\eqref{eq:PWAsysUYaug}, we have	
	\begin{multline}
	\label{eq:intSij_dAS}
		\!\!\bx(t_b)^T\bP_\ij\bx(t_b) - \bx(t_a)^T\bP_\ij\bx(t_a) + \phantom{a}\\ 
		\int_{t_a}^{t_b}\!\big(\norm{\Delta\qtau}^2 - \eta^{-2}\norm{\Delta p_\epsilon(\tau)}^2\big)\,d\tau \\[-7pt] 
		\leq -\int_{t_a}^{t_b}\! \sigma_3\norm{\Delta\xtau}^2 \,d\tau
	\end{multline}
	with $\Delta x := x - \tx$, and $\Delta q$ and $\Delta p_\epsilon$ similarly defined. The same reasoning can be applied to the last inequality in~\eqref{eq:ContThm1}, post and pre multiplying by $\col(x - \tx,p_\epsilon - \tp_\epsilon)^T$ and $\col(x - \tx,p_\epsilon - \tp_\epsilon)$, which yields 
	\begin{multline}
	\label{eq:intSii_dAS}
		\Delta x(t_b)^TP_i\Delta x(t_b) - \Delta x(t_a)^TP_i\Delta x(t_a) + \phantom{a}\\ 
		\int_{t_a}^{t_b}\!\big(\norm{\Delta\qtau}^2 - \eta^{-2}\norm{\Delta p_\epsilon(\tau)}^2\big)\,d\tau \\[-7pt] 
		\leq -\int_{t_a}^{t_b}\! \sigma_3\norm{\Delta\xtau}^2 \,d\tau.
	\end{multline}
We note that the first terms in~\eqref{eq:intSij_dAS} and~\eqref{eq:intSii_dAS} represent the incremental Lyapunov function~\eqref{eq:PWQStorFunc}. Let us consider a trajectory $\bxtau$, $\forall \tau \in [0,t]$. The time $t_1$ can be decomposed as $t = t - t_{in,n} + \sum_{k=0}^{n-1} (t_{out,k} - t_{in,k})$, with $t_{out,k} = t_{in,k+1}$ and $t_{in,0} = 0$, so that during each time interval $[t_{in,k},t_{out,k}]$ the trajectory stays in a given region. Then, replacing $t_a$ by $t_{in,k}$ and $t_b$ by $t_{out,k}$ in~\eqref{eq:intSij_dAS} and~\eqref{eq:intSii_dAS}, adding up to $n$ for every region $X_\ij$ crossed, and using the continuity of $V$ yields
	\begin{multline}
		V(\xt,\txt) - V(\xo,\txo) + \phantom{0}\\
		\int_{0}^{t}\! \big(\norm{\Delta\qtau}^2 - \eta^{-2}\norm{\Delta p_\epsilon(\tau)}^2\big)\,d\tau \\[-7pt] 
		\leq -\int_{0}^{t}\! \sigma_3\norm{\Delta\xtau}^2 \,d\tau.
	\end{multline}
	 
	Since $\epsilon$ is Lipschitz with a Lipschitz constant equal to $\eta$, the quantity $\norm{\Delta q}^2 - \eta^{-2}\norm{\Delta p_\epsilon}^2$ is always positive, and we obtain
	\begin{equation}
	\label{eq:intV}
		V(\xt,\txt) - V(\xo,\txo) \leq -\!\int_{0}^{t_1}\! \sigma_3\norm{\Delta\qtau}^2 \,d\tau.
	\end{equation}
	
	Then $V$ satisfies the conditions in Theorem~\ref{th:dAS} with $\alpha_i(r) := \sigma_i\norm{r}^2$, for $i \in \{1,2\}$ and $\rho(r) := \sigma_3\norm{r}^2$. The function $V$ is then an incremental Lyapunov function and system~\eqref{eq:LurePWAApp} is incrementally asymptotically stable, which concludes the proof.\hfill$\Box$
\end{proof}

Theorem~\ref{th:LurePWA_dAS} is of independent interest, as it extends the incremental circle criterion to the framework of PWA Lur'e systems. Indeed, by taking $N = 1$, we recover the LMI conditions of the classic incremental circle criterion (see e.g.~\cite{Fromion1999}).

In the proof of Theorem~\ref{th:LurePWA_dAS}, we construct an incremental Lyapunov function that ensures incremental asymptotic stability. Another interpretation can be given in view of the framework of dissipative systems~\citep{Willems1972}. Indeed, Theorem~\ref{th:LurePWA_dAS} can be seen as an incremental small gain theorem between the PWA system and the Lipschitz nonlinearity, where $V$ would play the role of the storage function, with supply rate $w(\bq,\bp,\bx) = \eta^{-2}\norm{p - \tp}^2 - \norm{q - \tq}^2 - \sigma_3\norm{x-\tx}^2$.

\section{Piecewise-affine approximation of scalar nonlinearities}
\label{se:AppNL}

Let us define $\Phi(N)$ as the set of piecewise-affine functions $\vpwa:\R \rightarrow \R$ defined on a partition of size $N$. That is, $\Phi(N)$ is the set of piecewise-affine functions for which there exists a partition $\lbrace\mRi\rbrace_{i\in\I}$ of $\R$, with $\abs{\I} = N$. Then, $\vpwa(q) = r_iq + s_i$, for $q \in \mRi$, where $i \in \I = \lbrace 1, \ldots, N\rbrace$. Since $\varphi$ is continuous and $\epsilon$ is Lipschitz continuous, $\vpwa$ must be continuous. This implies that $\forall q \in \mRi \cap \mR_j$, $r_iq + s_i = r_jq + s_j$. We also fix $\vpwa(0) = 0$, and then whenever $q = 0 \in \mRi$, we have $s_i = 0$. We shall make the following assumption on the nonlinearity $\varphi$.

\begin{assumption}
\label{ass:Varphi}
	The memoryless nonlinearity $\varphi$ is continuously differentiable, i.e. $\varphi \in \mathcal{C}^1(\R)$,  and asymptotically linear, i.e. there exist $k_1,k_2 \in \R$ such that $\lim_{q\rightarrow-\infty} \abs{\varphi'(q) - k_1} = 0$ and $\lim_{q\rightarrow\infty} \abs{\varphi'(q) - k_2} = 0$.
\end{assumption}

Assumption~\ref{ass:Varphi} ensures that we are able to construct an approximation $\vpwa$ with a finite partition, i.e. with $N < \infty$. We are interested in finding $\vpwa$ that best approximates $\varphi$. We shall measure the approximation error by its Lipschitz constant, i.e., by its incremental gain. This may be formalized as

\begin{equation}
\label{eq:OptProb}
	\begin{aligned}
		& \underset{\mathclap{\vpwa \in \Phi(N)}}{\text{minimize}}
		& & \eta \\
		& \text{subject to}
		& & \norm{\epsilon(q) - \epsilon(\tq)} \leq \eta\norm{q - \tq} \\
		&&& q, \tq \in \R
	\end{aligned} \tag{P1}
\end{equation}

As we refine the partition $\{\mRi\}_{i \in \I}$, by choosing a larger $N$, the approximation error decreases, while the complexity of $\vpwa$ increases. This indicates a trade-off between the accuracy of the description and the complexity of the analysis. We shall search for a value of $N$ ensuring a given upper bound $\etaref$ on the Lipschitz constant of the approximation error. This allows us to apply Theorem~\ref{th:LurePWA_dAS} to assess the incremental asymptotic stability of~\eqref{eq:Lure}. The next proposition gives a method to obtain $\vpwa$ respecting the desired upper bound on the approximation.

\begin{proposition}
\label{prop:UpBound}
	Let $\varphi$ be a function satisfying Assumption~\ref{ass:Varphi}. Let $\etaref > 0$, and let $\{\mRi\}_{i\in\I}$, with $\I = \lbrace 1,\ldots,N\rbrace$, be a partition of $\R$ obtained by a uniform division of the image of $\varphi'$ under $\R$, i.e. $l(\varphi'(\mRi)) = l(\varphi'(\mR_j))$, for all $i,j \in \I$, where $l(\cdot)$ denotes the length of an interval. Also, let $r_i = (\sup_{q \in \inte\mRi}\varphi'(q) + \inf_{q \in \inte\mRi}\varphi'(q))/2$ and $s_i$ be chosen to ensure continuity of $\vpwa$. Then, by choosing $N$ such that $l(\varphi'(\mRi)) \leq 2\etaref$, the obtained approximation $\vpwa$ ensures that $\epsilon$ is Lipschitz with a Lipschitz constant $\eta \leq \etaref$.
\end{proposition}
\begin{proof}
	We first use the fact that Lipschitz continuity is equivalent to boundedness of the derivative, for almost every $q \in \R$. Then, we show that the proposed partition method ensures the desired upper bound on the Lipschitz constant.
	
	We begin by recalling a known fact about Lipschitz functions. Let $\eta > 0$. For an arbitrary partition $\{\mRi\}_{i\in\I}$, the following two statements are equivalent:
	\begin{samepage}
	\begin{enumerate}[(i)]
		\item $\norm{\epsilon(q) - \epsilon(\tq)} \leq \eta\norm{q - \tq}$, for all $q,\tq \in \R$.
		\item $\norm{\epsilon'(q)} \leq \eta$, for almost all $q \in \R$.
	\end{enumerate}
	\end{samepage}
	We recall that $\epsilon'(q) = \varphi'(q) - r_i$, for all $q \in \inte\mRi$. Let $\eta_i > 0$ be such that $\sup_{q \in \inte\mRi} \norm{\epsilon'(q)} \leq \eta_i$. By choosing $r_i = (\sup_{q \in \inte\mRi}\varphi'(q) + \inf_{q \in \inte\mRi}\varphi'(q))/2$, we ensure that $\eta_i = l(\varphi'(\mRi))/2$. Since $\varphi$ is Lipschitz continuous, its derivative is bounded on $\R$. Then, we can use the proposed partition so that the image of $\varphi'$ under $\R$ is uniformly divided, and we have $\eta_i = l(\varphi'(\mRi))/2 \leq \etaref$, $\forall i \in \I$. Then, by defining $\eta = \eta_i$ and using the equivalent statements in the beginning of the proof, we have that $\norm{\epsilon(q) - \epsilon(\tq)} \leq \eta\norm{q - \tq}$, for all $q, \tq \in \R$, with $\eta \leq \etaref$, which concludes the proof. \hfill $\Box$
\end{proof}

The regions $\mRi = [q_i,q_{i+1}]$ can be defined by solving scalar nonlinear equations, which can be done by standard techniques such as the bisection method. We remark that, since $\varphi$ is asymptotically linear, the leftmost and rightmost regions $\mRi$ may be unbounded.

One could wonder whether the partition method in Proposition~\ref{prop:UpBound} gives the optimal solution to~\eqref{eq:OptProb}. It turns out that this is true, provided that $\varphi$ satisfies some new assumptions, as stated in the following.

\begin{assumption}
\label{ass:VarphiOpt}
	The memoryless nonlinearity $\varphi$ is odd, monotone, and so that $\varphi'$ is nondecreasing on $\R_+$.
\end{assumption}

\begin{proposition}
\label{prop:OptPart}
	Let $\varphi$ be a nonlinear function respecting Assumptions~\ref{ass:Varphi} and~\ref{ass:VarphiOpt}. Then, the partition method described in Proposition~\ref{prop:UpBound} yields $\vpwa$ that is the optimal solution to~\eqref{eq:OptProb}. 
\end{proposition}
\begin{proof}
	Due to the oddness of $\varphi$, we can focus on $\R_+$ and obtain the remaining by symmetry. Let $\{\mRi\}_{i \in \I}$ be an arbitrary partition of $\R_+$, with $\I = \{0,\ldots,m\}$. Also, let $\eta_i > 0$ be as in Proposition~\ref{prop:UpBound}. Then, by taking $\eta := \max_{i\in\I} \eta_i$, we have that $\norm{\epsilon'(q)} \leq \eta$, for almost all $q \in \R_+$. It is clear that, for each region, the choice of $r_i$ that minimizes $\eta_i$ is given by $r_i = (\sup_{q \in \inte\mRi}\varphi'(q) + \inf_{q \in \inte\mRi}\varphi'(q))/2$. In this case, we have $\eta _i = (\sup_{q \in \inte\mRi}\varphi'(q) - \inf_{q \in \inte\mRi}\varphi'(q))/2$. As $\varphi$ is Lipschitz, $\varphi'$ is bounded on $\R_+$. Since the derivative $\varphi'$ is continuous and nondecreasing on $\R_+$, we have 
	\begin{align}
	 	\sum_{i=0}^{m} \eta_i &= \sum_{i = 0}^{m} \frac{\sup_{q \in \inte{\mRi}}\varphi'(q) - \inf_{q \in \inte\mRi}\varphi'(q)}{2} \nonumber \\
	 	&= \frac{\ell(\varphi'(\R_+))}{2}.
	\end{align}
	From this, we are interested in minimizing $\eta = \max_{i\in\I} \eta_i$, subject to $\eta_i \geq 0$ and $\sum_{i=0}^{m} \eta_i=\ell(\varphi'(\R_+))/2$. The minimum is obtained when all $\eta_i$ have the same value, which is obtained by taking a partition such that the image of $\varphi'$ under $\R_+$ is uniformly divided. This yields $\eta = \ell(\varphi'(\R_+))/(2(m+1))$. Then, proceeding as in Proposition~\ref{prop:UpBound}, we conclude that $\vpwa$ obtained by this method ensures that $\norm{\epsilon(q) - \epsilon(\tq)} \leq \eta\norm{q - \tq}$, for all $q,\tq \in \R$, with $\eta$ minimal. \hfill$\Box$
\end{proof}

Despite the fact that Problem~\eqref{eq:OptProb} is non-convex due to the need to define the partition $\{\mRi\}_{i \in \I}$, Proposition~\ref{prop:OptPart} shows that, in the case where $\varphi$ satisfies Assumption~\ref{ass:VarphiOpt}, the optimal solution is known and quite easy to compute. The partitioning strategy is illustrated in Fig.~\ref{fig:PhiPartition}.
\begin{figure}
	\centering
		\includegraphics[width=0.5\linewidth]{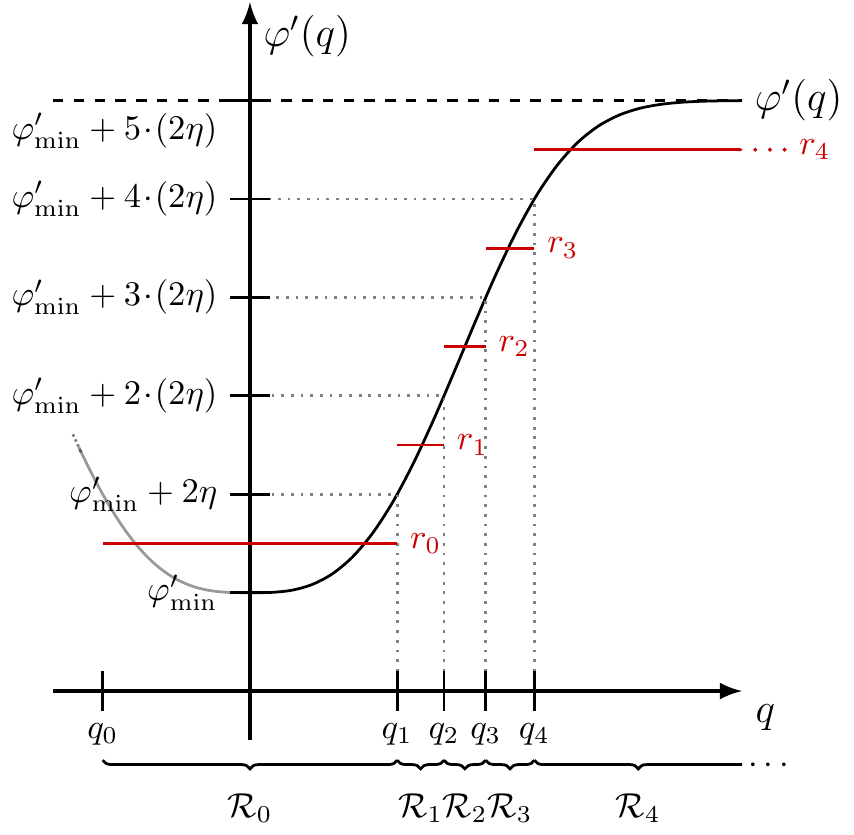}
	\caption{Partitioning strategy presented in Proposition~\ref{prop:OptPart}, based on the uniform division of the image of $\varphi'$ under $\R$.}
	\label{fig:PhiPartition}
\end{figure} 
In this case, we may explicitly compute $N$ such that the error bound is guaranteed to be inferior to $\eta_\textnormal{ref}$, as stated in the next proposition.

\begin{proposition}
\label{prop:Nbound}
	Let $\varphi$ be a nonlinearity satisfying Assumptions~\ref{ass:Varphi} and~\ref{ass:VarphiOpt}. Let $\etaref > 0$ be the desired upper bound on the Lipschitz constant of the approximation error. Then, if
	\begin{equation}
	\label{eq:NDef}
		m = \left\lceil \frac{\ell(\varphi'(\R_+))}{2\etaref}\right\rceil -1 > 0,
	\end{equation}
	with $N := 2m + 1$, and $\vpwa$ is obtained by the method in Proposition~\ref{prop:OptPart}, then the approximation error $\epsilon$ is Lipschitz with a Lipschitz constant $\eta \leq \etaref$.
\end{proposition}
\begin{proof}
	This is a simple consequence of the fact that the partitioning strategy presented in Proposition~\ref{prop:OptPart} ensures that $\eta = \ell(\varphi'(\R_+))/(2(m+1))$. \hfill$\Box$
\end{proof}

With the techniques presented in this section and Section~\ref{se:IncStabPWALure}, we have all the tools to apply Algorithm~\ref{alg:Method} to the study of the incremental asymptotic stability of Lur'e systems~\eqref{eq:Lure}. This shall be illustrated in the next section through some numerical examples.

\section{Numerical examples}
\label{se:NumEx}

\addtolength{\textheight}{-1cm}

\begin{example}
\label{ex:ExAcad}
	Consider the nonlinear system given by~\eqref{eq:Lure} with
	\begin{equation}
	 	\begin{aligned}
	 		A &= \nmatrix{cc}{-1 & 0 \\ 3 & -2} & B &= \nmatrix{c}{1 \\ 0} & C &= \nmatrix{cc}{0 & 1}
	 	\end{aligned}
	\end{equation}
	and $\varphi(q) = 2q^3$, for $\abs{q} \leq 1$, and $\varphi(q) = 6q - 4\sign(q)$, for $\abs{q} > 1$. $\varphi$ satisfies Assumption~\ref{ass:Varphi}, and belongs to the incremental sector $[0,6]$. Analysis via the incremental circle criterion does not lead to a conclusion on the incremental stability of the system. We aim to obtain a piecewise-affine approximation $\vpwa$ over $\R$, so that we can apply Theorem~\ref{th:LurePWA_dAS}. Let us fix the desired maximal Lipschitz constant as $\etaref = 0.8$. Using the approach proposed in Section~\ref{se:AppNL}, we obtain the approximation illustrated in Fig.~\ref{fig:ExAcadPWA}, with $N = 7$ and $\eta = 0.75$. Using Proposition~\ref{prop:LurePWA}, the system is transformed in the interconnection of a PWA system and a Lipschitz nonlinearity. We then successfully apply Theorem~\ref{th:LurePWA_dAS} to construct a piecewise-affine incremental Lyapunov function, and conclude that this system is globally incrementally asymptotically stable.	
	\begin{figure}[tb]
		\begin{subfigure}[t]{0.49\linewidth}
			\centering
			\resizebox{0.8\linewidth}{!}{
			\includegraphics{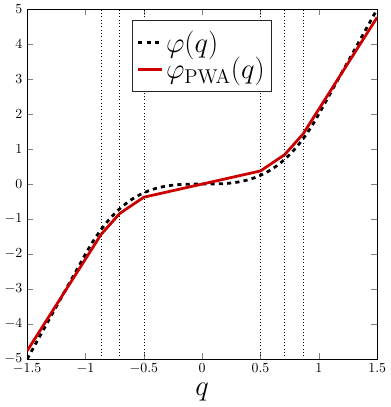}
			}
		\end{subfigure}
		\hfill
		\begin{subfigure}[t]{0.49\linewidth}
			\centering
			\resizebox{0.76\linewidth}{!}{
				\includegraphics{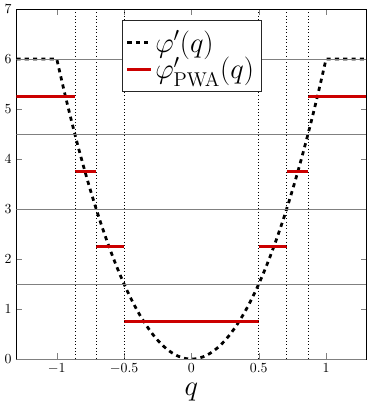}
			}
		\end{subfigure}
		\caption{Nonlinear function $\varphi(\alpha) = 2q^3$ in Example~\ref{ex:ExAcad} and the piecewise-affine approximation $\vpwa$. The dotted lines represent the partition $\{\mRi\}_{i\in\I}$.}
		\label{fig:ExAcadPWA}
	\end{figure}
\end{example}

\begin{example}
\label{ex:ExMissile}
	Let us consider the nonlinear missile benchmark presented in~\cite{Reichert1992}. The incremental behavior of the closed-loop system with a PI controller has been previously studied in~\cite{Fromion1999}. In this reference, the closed-loop system is written as an LTI system fedback through a nonlinearity $\varphi(\alpha) = -(a_n\alpha^3 + b_n\abs{\alpha}\alpha)$, with $\alpha$ being the angle of attack (see~\cite{Fromion1999} for complete model and details). This model is assumed to be valid for $\abs{\alpha}$ less than $20^\circ$ (or $0.34~\textnormal{rad}$). Using again the techniques in the previous section with $\etaref = 3.5$, we obtain the approximation $\vpwa$ presented in Fig.~\ref{fig:ExMissilePWA}, with $N = 5$ and $\eta = 2.4293$. Application of Theorem~\ref{th:LurePWA_dAS} allows us to assess the incremental asymptotic stability of the closed-loop system, which concurs with the observations on~\cite{Fromion1999} about the good behavior provided by the PI controller.
	\begin{figure}[tb]
		\centering
		\includegraphics[width=0.4\linewidth]{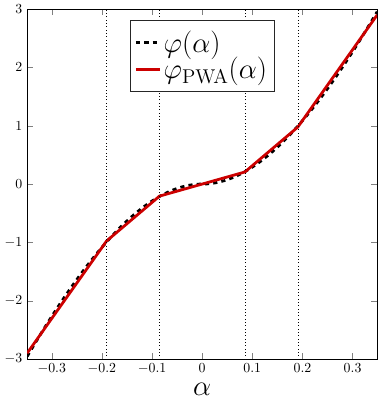}
		\caption{Nonlinear function $\varphi(\alpha) = -(a_n\alpha^3 + b_n\abs{\alpha}\alpha)$ in Example~\ref{ex:ExMissile} and the piecewise-affine approximation $\vpwa$.}
		\label{fig:ExMissilePWA}
	\end{figure}
\end{example}

\section{Conclusion}
\label{se:Conc}

In this \paper we have proposed a new method to assess incremental asymptotic stability of Lur'e systems, based on piecewise-affine approximations. As a byproduct, we extended the celebrated incremental circle criterion to the analysis of PWA Lur'e systems, with conditions that can be solved very efficiently by interior point solvers.

Perspectives for future work include the extension of the approach in Section~\ref{se:AppNL} to the case of multivariable nonlinearities, and the establishment of local results, e.g. in the case when the nonlinearity $\varphi$ is not asymptotically linear and a global approximation $\vpwa$ with a finite partition is not possible. Finally, the results in Section~\ref{se:AppNL} may be coupled with robustness analysis to ensure robust incremental stability of Lur'e systems. 

\bibliographystyle{abbrvnat}
\bibliography{Ref_arXiv2017}

\begin{thebibliography}{19}
\providecommand{\natexlab}[1]{#1}
\providecommand{\url}[1]{\texttt{#1}}
\expandafter\ifx\csname urlstyle\endcsname\relax
  \providecommand{\doi}[1]{doi: #1}\else
  \providecommand{\doi}{doi: \begingroup \urlstyle{rm}\Url}\fi

\bibitem[Angeli(2002)]{Angeli2002}
D.~Angeli.
\newblock A {L}yapunov approach to incremental stability properties.
\newblock \emph{{IEEE} Transactions on Automatic Control}, 47\penalty0
  (3):\penalty0 410--421, 2002.

\bibitem[Azuma et~al.(2010)Azuma, Imura, and Sugie]{Azuma2010}
S.~Azuma, J.~Imura, and T.~Sugie.
\newblock {L}ebesgue piecewise affine approximation of nonlinear systems.
\newblock \emph{Nonlinear Analysis: Hybrid Systems}, 4\penalty0 (1):\penalty0
  92 -- 102, 2010.
\newblock ISSN 1751-570X.

\bibitem[Fromion(1997)]{Fromion1997}
V.~Fromion.
\newblock Some results on the behavior of {L}ipschitz continuous systems.
\newblock In \emph{European Control Conference (ECC)}, pages 2011--2016,
  Brussels, Belgium, July 1997.

\bibitem[Fromion and Safonov(2004)]{Fromion2004a}
V.~Fromion and M.~G. Safonov.
\newblock {P}opov-{Z}ames-{F}alb multipliers and continuity of the input/output
  map.
\newblock In \emph{{IFAC} Symposium on Nonlinear Control Systems ({NOLCOS})},
  Stuttgart, Germany, 2004.

\bibitem[Fromion et~al.(1999)Fromion, Scorletti, and Ferreres]{Fromion1999}
V.~Fromion, G.~Scorletti, and G.~Ferreres.
\newblock Nonlinear performance of a {PI} controlled missile: an explanation.
\newblock \emph{International Journal of Robust and Nonlinear Control},
  9\penalty0 (8):\penalty0 485--518, 1999.

\bibitem[Fromion et~al.(2003)Fromion, Safonov, and Scorletti]{Fromion2003a}
V.~Fromion, M.~G. Safonov, and G.~Scorletti.
\newblock Necessary and sufficient conditions for {L}ur'e system incremental
  stability.
\newblock In \emph{European Control Conference ({ECC})}, pages 71--76,
  Cambridge, United Kingdom, Sept 2003.

\bibitem[Johansson and Rantzer(1998)]{Johansson1998}
M.~Johansson and A.~Rantzer.
\newblock Computation of piecewise quadratic {L}yapunov functions for hybrid
  systems.
\newblock \emph{{IEEE} Transactions on Automatic Control}, 43\penalty0
  (4):\penalty0 555--559, 1998.

\bibitem[Kim and de~Persis(2015)]{Kim2015}
H.~Kim and C.~de~Persis.
\newblock Output synchronization of {L}ur'e-type nonlinear systems in the
  presence of input disturbances.
\newblock In \emph{{IEEE} Conference on Decision and Control (CDC)}, pages
  4145--4150, Osaka, Japan, Dec 2015.

\bibitem[Kulkarni and Safonov(2002)]{Kulkarni2002}
V.~V. Kulkarni and M.~G. Safonov.
\newblock Incremental positivity nonpreservation by stability multipliers.
\newblock \emph{{IEEE} Transactions on Automatic Control}, 47\penalty0
  (1):\penalty0 173--177, Jan 2002.
\newblock ISSN 0018-9286.

\bibitem[Liberzon(2006)]{Liberzon2006}
M.~R. Liberzon.
\newblock Essays on the absolute stability theory.
\newblock \emph{Automation and Remote Control}, 67\penalty0 (10):\penalty0
  1610--1644, 2006.

\bibitem[Lohmiller and Slotine(1998)]{Lohmiller1998}
W.~Lohmiller and J.-J.~E. Slotine.
\newblock On contraction analysis for non-linear systems.
\newblock \emph{Automatica}, 34\penalty0 (6):\penalty0 683--696, 1998.
\newblock ISSN 0005-1098.

\bibitem[Pavlov et~al.(2004)Pavlov, Pogromsky, van~de Wouw, and
  Nijmeijer]{Pavlov2004}
A.~Pavlov, A.~Pogromsky, N.~van~de Wouw, and H.~Nijmeijer.
\newblock Convergent dynamics, a tribute to {B}oris {P}avlovich {D}emidovich.
\newblock \emph{Systems \& Control Letters}, 52\penalty0 (3–4):\penalty0 257
  -- 261, 2004.
\newblock ISSN 0167-6911.

\bibitem[Rantzer(2000)]{Rantzer2000a}
A.~Rantzer.
\newblock A performance criterion for anti-windup compensators.
\newblock \emph{European Journal of Control}, 6\penalty0 (5):\penalty0
  449--452, 2000.
\newblock ISSN 0947-3580.

\bibitem[Reichert(1992)]{Reichert1992}
R.~T. Reichert.
\newblock Dynamic scheduling of modern-robust-control autopilot designs for
  missiles.
\newblock \emph{{IEEE} Control Systems}, 12\penalty0 (5):\penalty0 35--42, Oct
  1992.
\newblock ISSN 1066-033X.

\bibitem[Waitman et~al.(2016)Waitman, Massioni, Bako, Scorletti, and
  Fromion]{Waitman2016}
S.~Waitman, P.~Massioni, L.~Bako, G.~Scorletti, and V.~Fromion.
\newblock Incremental $\mathcal{L}_2$-gain analysis of piecewise-affine systems
  using piecewise quadratic storage functions.
\newblock In \emph{{IEEE} Conference on Decision and Control}, Las Vegas, USA,
  2016.

\bibitem[Willems(1972)]{Willems1972}
J.~C. Willems.
\newblock Dissipative dynamical systems parts {I} and {II}.
\newblock \emph{Archive for Rational Mechanics and Analysis}, 45\penalty0
  (5):\penalty0 321--393, 1972.

\bibitem[Zames(1966)]{Zames1966}
G.~Zames.
\newblock On the input-output stability of time-varying nonlinear feedback
  systems---parts {I} and {II}.
\newblock \emph{{IEEE} Transactions on Automatic Control}, 11\penalty0
  (2):\penalty0 228--238, 465--476, 1966.

\bibitem[Zames and Falb(1968)]{Zames1968a}
G.~Zames and P.~L. Falb.
\newblock Stability conditions for systems with monotone and slope-restricted
  nonlinearities.
\newblock \emph{{SIAM} Journal on Control}, 6\penalty0 (1):\penalty0 89--108,
  1968.

\bibitem[Zavieh and Rodrigues(2013)]{Zavieh2013}
A.~Zavieh and L.~Rodrigues.
\newblock Intersection-based piecewise affine approximation of nonlinear
  systems.
\newblock In \emph{Mediterranean Conference on Control Automation (MED)}, pages
  640--645, Platanias-Chania, Greece, June 2013.

\end{thebibliography}

\end{document}